\documentclass{ws-ijgmmp}
\begin{document}
\newcommand{\G}{{\cal G}}

\markboth{ S.Capozziello, C.A. Mantica, L.G. Molinari}
{Gauss-Bonnet perfect fluids}

%
\catchline{}{}{}{}{}
%

\title{COSMOLOGICAL PERFECT FLUIDS \\IN GAUSS-BONNET GRAVITY}

\author{SALVATORE CAPOZZIELLO}
\address{
Dipartimento di Fisica ``E. Pancini",  Universit\`a degli Studi di Napoli
	\textquotedblleft{Federico II}\textquotedblright, and\\
INFN Sez. di Napoli, Compl. Univ. di Monte S. Angelo, Edificio G, Via Cinthia, I-80126, Napoli, Italy,\\
Tomsk State Pedagogical University, ul. Kievskaya, 60, 634061 Tomsk, Russia.\\
\email{capozziello@na.infn.it}
}

\author{CARLO ALBERTO MANTICA}
\address{I.I.S. Lagrange, Via L. Modignani 65, 
I-20161, Milano, Italy \\
and INFN sez. di Milano,
Via Celoria 16, I-20133 Milano, Italy\\
\email{carlo.mantica@mi.infn.it}
}

\author{LUCA GUIDO MOLINARI}
\address{Dipartimento di Fisica ``A. Pontremoli'',
Universit\`a degli Studi di Milano\\ and INFN sez. di Milano,
Via Celoria 16, I-20133 Milano, Italy\\
\email{luca.molinari@unimi.it}
}

\maketitle

\begin{history}
\received{(Day Month Year)}
\revised{(Day Month Year)}
\end{history}

\begin{abstract}
In a $n$-dimensional Friedmann-Robertson-Walker metric, it is rigorously shown that any 
analytical theory of gravity $f(R,\G)$, where $R$ is the curvature scalar and $\G$ is
the Gauss-Bonnet topological invariant, can be associated to a perfect-fluid stress-energy tensor.
In this perspective, dark components of the cosmological Hubble flow can be geometrically interpreted.
\end{abstract}

\section{Introduction}
The big puzzle of modern cosmology is the fact that most of the 95\% of the {\it cosmic pie} is composed by dark energy and dark matter. Despite the huge experimental efforts to detect these elusive components at fundamental level, there is no final evidence, at the moment, confirming their particle nature.
On the other hand, clustering phenomena at any galactic and extragalactic scale and accelerated expansion of the Hubble flow point out  that the  paradigm based on the General Relativity and the Standard Model of Particles is not sufficient to explain the  observations. According to this state of art, it is mandatory to search for alternative explanations in order to address the cosmological phenomenology unless incontrovertible tests for new particles are found. 

Specifically, observations indicate a spatially flat homogeneous and isotropic universe on  large scale. This  property  is described
by the Friedmann-Robertson-Walker (FRW) metric. The observed facts require that the  Einstein equations 
$ R_{ij} - \tfrac{1}{2}  g_{ij} R= \kappa T_{ij} $ 
contain a stress-energy tensor $T_{ij}=T_{ij}^{matt} + T_{ij}^{dark}$ of a perfect fluid that accounts for ordinary matter as well as
 dark matter and dark energy. 
However, the  alternative road consists in explaining  the dark fluid $T_{ij}^{dark}$ as a geometrical effect. In other words, one can attempt at reproducing the observed facts by modifying the left-hand side
of the equations of gravity, within the appropriate frame of FRW space-times. This is done by replacing
the Hilbert-Einstein action, linear in the Ricci scalar $R$,  with   approaches where more general curvature  
 \cite{CF08,CDL11,Nojiri17} or torsion  \cite{Manos} invariants are considered into the dynamics. The so called $f(R)$ gravity has been the first example in this direction \cite{Cap}.

With this premise, it has been proved that the stress-energy tensor has the perfect-fluid form
in any $f(R)$ cosmology  in generalized FRW space-times with harmonic Weyl tensor \cite{CMM2018}. This means that the additional  terms, produced by extending geometry (that is from $R$ to functions  $f(R)$, non-linear in $R$), sum up to a
perfect-fluid tensor that may well fit as a source term in the right-hand side of the Einstein equations.\\
In this paper, the result is generalized to  $f (R,\G)$  theories of gravitation in FRW space-times, where $f (R,\G)$ is an 
arbitrary  analytical function of the scalar curvature $R$ and of the Gauss-Bonnet topological invariant 
\begin{align}\label{2.2}
\G=R_{jklm}R^{jklm} -4 R_{jk}R^{jk} +R^2\,.
\end{align}
Modified gravity with Gauss-Bonnet scalar in the form $R+f(\G)$ was first introduced   in the context of FRW metric, as alternative to dark energy for the late acceleration of the universe \cite{Nojiri05}. Several other investigations followed \cite{Cognola06,Santos18,Benetti18, MF, Odintsov17} because the Gauss-Bonnet term results useful to regularize the gravitational theory for quantum fields in curved spaces \cite{CDL11} and  improves the efficiency of inflation giving rise to multiple accelerated expansions \cite{Paolella} because $\G$ behaves as a further {\it scalaron} besides $R$ \cite{Starobinsky80}.  

In $n$ dimensions the gravitational action  is  
\begin{align}\label{2.1}
S=\frac{1}{2\kappa}\int d^nx \sqrt{-g} f(R,{\G}) +S^{matt} ,
\end{align}
where $S^{matt}$ is the action term of standard matter fields. The first variation in the metric gives the field equations
\begin{align}\label{2.3}
R_{kl} -\tfrac{1}{2}  g_{kl} R= \kappa T^{matt}_{kl} + \Sigma_{kl}
\end{align}
where $T^{matt}_{kl} $ results from $S^{matt}$ and the tensor $\Sigma_{kl}$ arises from the geometry.  
Assuming $f_R = \partial_R f $ and $f_{\G} =\partial_\G f $, the latter is \cite{Atazadeh14}:
\begin{align}
\Sigma_{kl} =&( \nabla_k\nabla_l  -g_{kl}\nabla^2) f_R  + 2R (\nabla_k\nabla_l  - g_{kl}\nabla^2) f_\G -4(R_k{}^m\nabla_m\nabla_l +R_l{}^m\nabla_m\nabla_k)f_\G \nonumber\\
& +4 (R_{kl}\nabla^2  + g_{kl} R^{pq}\nabla_p\nabla_q  +  R_{kpql}\nabla^p\nabla^q) f_\G     -\tfrac{1}{2}g_{kl}(R f_R+\G f_\G -  f)       \label{Sigma} \\
& + (1-f_R) (R_{kl}-\tfrac{1}{2} g_{kl}R) \nonumber
\end{align}
Despite of the complexity of the expression, we prove the following:
\begin{theorem}
In a Friedmann-Robertson-Walker space-time of dimension $n$, for any analytical $f(R,\G)$ model of gravity, 
the tensor $\Sigma_{kl}$ is a  perfect fluid of the form
\begin{align} \frac{1}{\kappa}\Sigma_{kl}= (p+\mu) u_ku_l+p g_{kl} 
\label{perfect}
\end{align}
\end{theorem}
\noindent
The scalar fields $p$ and $\mu$ can be  interpreted as  effective  pressure and energy density produced by 
the geometry in the locally comoving frame ($u^0=1$).
Their explicit expressions will be  given in Eq.\eqref{expsigma} and following  ones. 

We have adopted the following notations: for a scalar $S$ we use  $\dot S = u^m\nabla_m S $, for a vector $v^k$ we write 
$v^2$ for $v^k v_k$, and $\nabla^2$ for $\nabla^k\nabla_k$. The  signature is ($-,+,\dots,+$).
Before giving the proof of the above statement, let us discuss the covariant description of physical quantities in FRW space-times.
\section{Covariant description of FRW space-times}
The perfect fluid representation of $f(R,\G)$ gravity can be achieved  in the framework of a covariant description of FRW space-times. Based on a theorem by Chen \cite{Chen14},\cite{GRWSurv17}
a FRW space-time is
characterized by the Weyl tensor being zero and by the existence of a time-like vector field $u^k$ ($u^k u_k =-1$) such that \cite{[29]}
\begin{align} \nabla_j u_k =\varphi (u_ju_k + g_{jk}) \end{align}
where $\varphi  $ is a scalar field such that $\nabla_i \varphi = - u_i \dot\varphi $.
The scalar $\varphi $ coincides with
Hubble's parameter $H=\dot a/a$ in the comoving frame, where $a(t)$ is the scale factor of the FRW metric 
in its standard warped expression. By evaluating the Riemann  tensor
\begin{equation}
R_{jkl}{}^mu_m = [\nabla_j,\nabla_k]u_l = (\dot\varphi+\varphi^2)(u_k g_{jl} - u_j g_{kl})\,, 
\end{equation}
and using it in the equation for the Weyl tensor, $0=C_{jkl}{}^mu_m $, one obtains the Ricci tensor 
of a FRW space-time \cite{MaMoJMP16}. It has the perfect fluid structure being:
 \begin{align}
R_{kl} = \frac{R-n\xi}{n-1} u_k u_l + \frac{R-\xi}{n-1} g_{kl} \label{Ricci}\,,
\end{align}
with the Ricci scalar $R=R^k{}_k$. The eigenvalue is 
\begin{equation}
\xi = (n-1)(\dot \varphi + \varphi^2) =(n-1)\ddot a/a\,.
\end{equation} 
The curvature scalar is then 
\begin{align}
R = \frac{R^*}{a^2} + (n-1)(n-2)\varphi^2 +2\xi
\end{align}
where $R^*$ is the constant (in space-time) curvature of the space-like surfaces orthogonal to $u^k$.
It solves the equation resulting from the covariant derivative of $R_{kj}u^j=\xi u_k$, that is:
\begin{align}
\dot R -2\dot\xi = -2\varphi (R - n\xi)\,.  
\end{align} 
%
An application of this formalism to a toy-model in cosmology is done in \cite{toy}.

In \cite{CMM2018} we proved the following relations, valid for generalized FRW space-times with harmonic Weyl tensor and,
in particular, in FRW space-times:
\begin{gather}
\nabla_k R = -u_k \dot R, \qquad \nabla_k\nabla_l R =  -\varphi \dot R \, g_{kl} - (\varphi \dot R-\ddot R) u_k u_l\,, \label{nabR}\\
\nabla_k \xi = -u_k \dot \xi, \qquad \nabla_k\nabla_l \xi =  -\varphi \dot \xi \, g_{kl} - (\varphi \dot \xi-\ddot \xi) u_k u_l\,. \label{nabxi} 
\end{gather}

\begin{proposition} In a FRW space-time it is:
\begin{equation}
\nabla_k \G = -u_k \dot \G, \qquad \nabla_k\nabla_l \G =  -\varphi \dot \G \, g_{kl} - (\varphi \dot \G-\ddot \G) u_k u_l 
\end{equation}
\end{proposition}
\begin{proof}
It is advantageous to express the Gauss-Bonnet scalar $\G$ in terms of the Weyl scalar
\begin{equation}
R_{jklm}R^{jklm}=
C_{jklm}C^{jklm} +\frac{4}{n-2}R_{jk}R^{jk}-\frac{2}{(n-1)(n-2)}R^2\,,
\end{equation}
in the form
\begin{align*}
\G = C_{jklm}C^{jklm} - 4 \frac{n-3}{n-2} R^{kl}R_{kl} + \frac{n(n-3)}{(n-1)(n-2)} R^2 
\end{align*}
In a FRW space-time it is  $C_{jklm}=0$ then, by Eq.\eqref{Ricci}, we get:
 \begin{equation}
R_{kl} R^{kl}= \xi^2+\frac{1}{n-1} (R-\xi)^2\,.
\end{equation}
 The Gauss-Bonnet invariant reduces to
 \begin{equation}
\G =  \frac{n-3}{(n-1)(n-2)}[(n-4)R+2n\xi](R-2\xi)\,.
\end{equation}
It follows that
\begin{equation}
 \nabla_k \G = \tfrac{n-3}{(n-1)(n-2)}[2R(n-4)\nabla_k R -8n\xi \nabla_k \xi + 8\xi\nabla_k R  + 8 R \nabla_k \xi]\,.
 \end{equation}
By the properties \eqref{nabR} and \eqref{nabxi}, the first assertion is proven being
\begin{align}
\dot  \G = \tfrac{n-3}{(n-1)(n-2)}[(2nR -8R+8\xi)\dot R + 8(R-n\xi)\dot \xi ]\,.
\end{align}
The other assertion follows  by differentiation.
\end{proof}
These results allow to set  the tensor $\Sigma_{kl}$ in a perfect fluid form.  

\section{Evaluation of $\Sigma_{kl}$ and the perfect-fluid form}
For an analytical  function $f(R,\G)$ the double derivatives are evaluated as
\begin{align*}
\nabla_k\nabla_l f_R = 
&  f_{RRR} (\nabla_k R)(\nabla_l R) + f_{RR\G} [(\nabla_k R) (\nabla_l \G) + (\nabla_l R)(\nabla_k \G)]\\
&+ f_{R\G\G} (\nabla_k \G)(\nabla_l \G) + f_{RR} \nabla_k\nabla_l R + f_{R\G} \nabla_k\nabla_l \G \label{RR}\\
\nabla_k\nabla_l f_\G =&  f_{\G\G\G} (\nabla_k \G)(\nabla_l \G) + f_{R\G\G}[ (\nabla_k R)( \nabla_l \G) + (\nabla_l R)(\nabla_k \G)]\\
&+ f_{RR\G} (\nabla_k R)(\nabla_l R) + f_{\G\G} \nabla_k\nabla_l \G + f_{R\G} \nabla_k\nabla_l R\,
\end{align*}
In a Robertson-Walker space-time they gain the perfect-fluid form being:
\begin{align*}
\nabla_k\nabla_l f_R =&  -\varphi (g_{kl}+u_k u_l)(f_{RR} \dot R + f_{R\G} \dot \G)\\
&+u_k u_l (f_{RRR} \dot R \dot R  + 2f_{RR\G} \dot R \dot \G + f_{R\G\G} \dot \G\dot \G +f_{RR}\ddot R +f_{R\G}\ddot \G )  \\
\nabla_k\nabla_l f_\G =&  -\varphi (g_{kl}+u_ku_l) (f_{\G\G}\dot \G + f_{R\G} \dot R) \\
&+u_k u_l (f_{\G\G\G} \dot \G \dot \G +2 f_{R\G\G} \dot R \dot \G + f_{RR\G} \dot R\dot R + f_{R\G}\ddot R + f_{\G\G}\ddot \G )
\end{align*}
The perfect-fluid form is more evident by introducing the simplified expressions:
\begin{equation}
\nabla_k\nabla_l f_R = A_R g_{kl}+ B_R u_ku_l, \qquad 
\nabla_k\nabla_l f_\G = A_\G g_{kl} + B_\G u_ku_l \,.
\end{equation} 
In particular, we have  
\begin{equation}
\nabla^2 f_R=nA_R-B_R, \quad \mbox{and} \quad \nabla^2 f_\G=nA_\G-B_\G\,,
\end{equation}
that can be easily introduced in \eqref{Sigma}. That is 
\begin{align*}
\Sigma_{kl}=& (A_R g_{kl} + B_R u_ku_l) -g_{kl} (nA_R-B_R)\\
& + 2R (A_\G g_{kl} + B_\G u_ku_l) -2R g_{kl} (nA_\G-B_\G)\\
&-4[R_k{}^m(A_\G g_{ml}+B_\G u_m u_l)+R_l{}^m(A_\G g_{mk}+B_\G u_m u_k)]\\
& +4 R_{kl} (n A_\G-B_\G) +4g_{kl} R^{pq} (A_\G g_{pq}+B_\G u_pu_q)+4R_{kpql}(A_\G g^{pq}+B_\G u^pu^q)\\
& -\tfrac{1}{2}g_{kl}(R f_R+\G f_\G -  f) + (1-f_R) (R_{kl}-\tfrac{1}{2} g_{kl}R)
\end{align*}
Simplifying  with $R_{jk}u^k=\xi u_j$ and $R_{kpql}u^pu^q = \frac{1}{n-1}\xi(g_{kl}+u_ku_l)$, it is 
\begin{align*}
\Sigma_{kl}=& -g_{kl}(n-1)(A_R+2RA_\G) + (g_{kl}+u_ku_l )(B_R+2RB_\G)-8R_{kl}A_\G - 8 \xi B_\G u_ku_l\\
& +4 R_{kl} (nA_\G-B_\G) +4g_{kl}(RA_\G-\xi B_\G) -4 R_{kl}A_\G \\
&+\tfrac{4}{n-1}\xi B_\G(g_{kl}+u_ku_l) -\tfrac{1}{2}g_{kl}(\G f_\G -  f+R) + (1-f_R)R_{kl}\,.
\end{align*}
Finally, we obtain the  expression in a  perfect-fluid form as in \eqref{perfect}, that is:
\begin{align}\label{expsigma}
\Sigma_{kl} =& \frac{g_{kl}}{n-1} [-(n-1)^2A_R + (n-1) B_R +2(n-3)(R-2\xi)B_\G\nonumber \\
& -2(n-3)(R(n-3)+2\xi)A_\G -\tfrac{1}{2}(n-1)(\G f_\G-f+R)+(R-\xi)(1-f_R)]   \nonumber\\
&+\frac{u_ku_l}{n-1} [(n-1)B_R +2(n-3)(R-2\xi)B_\G + 4(n-3)(R-n\xi)A_\G  \\
&+(R-n\xi)(1-f_R)]\,.\nonumber
\end{align}
where:
\begin{align}
A_R = &-\varphi (f_{RR} \dot R + f_{R\G} \dot \G)\,,\\
B_R  =& f_{RRR} (\dot R)^2   + 2f_{RR\G} \dot R \dot \G + f_{R\G\G} (\dot \G )^2 -f_{RR}(\varphi \dot R -\ddot R) - f_{R\G}(\varphi \dot\G -\ddot \G)\,,\\
A_\G=&-\varphi (f_{\G\G} \dot \G + f_{R\G} \dot R)\,,\\
B_\G =& f_{\G\G\G} (\dot \G )^2 +2 f_{R\G\G} \dot R \dot \G + f_{RR\G} (\dot R)^2 - f_{R\G}(\varphi \dot R- \ddot R) -f_{\G\G}(\varphi \dot\G -\ddot \G)\,,
\end{align}
with $\dot A_R = \dot\varphi (A_R/\varphi ) - \varphi B_R$ and $\dot A_\G = \dot\varphi (A_\G/\varphi ) - \varphi B_\G$. 
The interpretation of the coefficients in Eq.\eqref{expsigma} is straightforward. The term in $g_{kl}$ is the effective pressure
\begin{align}\label{expsigma}
p =& \left(\frac{1}{n-1}\right) \left[-(n-1)^2A_R + (n-1) B_R +2(n-3)(R-2\xi)B_\G\right. \\
&\left. -2(n-3)(R(n-3)+2\xi)A_\G -\tfrac{1}{2}(n-1)(\G f_\G-f+R)+(R-\xi)(1-f_R)\right] \,, \nonumber 
\end{align}
and the term in the velocities $u_k u_l$ is the sum of the effective pressure and effective energy density
\begin{align}\label{expsigma}
p+\mu =&\left(\frac{1}{n-1}\right) \left[(n-1)B_R +2(n-3)(R-2\xi)B_\G + 4(n-3)(R-n\xi)A_\G\right.  \nonumber\\
&\left.+(R-n\xi)(1-f_R)\right]\,.
\end{align}

\section{Discussion and Conclusions}

In this paper, we  extended the results reported in \cite{CMM2018} where extra-contributions to the Einstein field equations, coming from $f(R)$ gravity, were recast as an effective perfect fluid of geometrical origin. Here we continued this ``geometrization" program by including contributions coming from the Gauss-Bonnet topological invariant $\G$ into the dynamics. This approach  improves the perspective of interpreting the dark side of the universe within the standard of geometric effects \cite{Cardo}. As for the previous case, we showed that the extra source produces a perfect-fluid term in the right hand side of the field equations. The addition of the Gauss-Bonnet invariant exhausts all 
possibilities of  fourth-order dynamics coming from curvature invariants. In fact, an action involving curvature invariants means to take into account combinations/functions of $R$, $R_{ik}$ and $R^k_{ilm}$. The inclusion of $\G$ fixes a relation among them and so $f(R,\G)$ is a general fourth-order action of gravity where all possible curvature invariants are taken into account. Clearly, the program can be enlarged by considering terms like $\Box R$,  non-local terms like $\Box^{-1} R$,  torsion invariants and so on. The perspective is to achieve a classification of perfect fluids coming from geometry. This ``geometric view"   could address the puzzle of dark side and explain why no particles beyond the Standard Model have yet been detected.  However, this statement has to be confirmed, at least, by a stringent  cosmographic analysis where possible components to the cosmic bulk  should be observationally discriminated \cite{OrlandoDunsby,ester1,ester2,revcosmo,aviles,gruber}. In a forthcoming paper, this problem will be considered.
 
\section*{Acknowledgments}
S. C.   acknowledges the support of  INFN ({\it iniziative specifiche} MOONLIGHT2 and QGSKY).
This paper is based upon work from COST action CA15117 (CANTATA), supported by COST (European 
Cooperation in Science and Technology).
%

%
\vfill
\end{document}